\documentclass[sigconf,screen]{acmart}

\usepackage{booktabs} % For formal tables
\usepackage{stmaryrd}
\usepackage{algorithm}
\usepackage{algpseudocode}
\usepackage{listings}
\usepackage{pifont}
\usepackage{stmaryrd}
\usepackage{multirow}
\usepackage{balance}
\usepackage{caption}
\usepackage{subcaption}
\usepackage{makecell}
\usepackage{tabulary}
\usepackage[absolute]{textpos}
\usepackage{tcolorbox}
\begin{document}
\title[DeepGini: Prioritizing Massive Tests to Enhance the Robustness of Deep Neural Networks]{DeepGini: Prioritizing Massive Tests to Enhance the Robustness of Deep Neural Networks}

\author{Yang Feng}
\affiliation{
	\institution{State Key Lab for Novel Software Technology, Nanjing University}
	\city{Nanjing}
	\country{China}
	}
\email{fengyang@nju.edu.cn}

\author{Qingkai Shi}
\affiliation{
	\institution{The Hong Kong University of Science and Technology}
	\city{Hong Kong}
	\country{China}
	}
\email{qshiaa@cse.ust.hk}

\author{Xinyu Gao}
\affiliation{
	\institution{State Key Lab for Novel Software Technology, Nanjing University}
	\city{Nanjing}
	\country{China}
}
\email{mf1932046@smail.nju.edu.cn}

\author{Jun Wan}
\affiliation{
	\institution{Ant Financial Services Group}
	\city{Hangzhou}
	\country{China}
}
\email{wukun.wj@antfin.com}

\author{Chunrong Fang}
\affiliation{%
	\institution{State Key Lab for Novel Software Technology, Nanjing University}
	\city{Nanjing}
	\country{China}
}
\email{fangchunrong@nju.edu.cn}

\author{Zhenyu Chen}
\affiliation{%
	\institution{State Key Lab for Novel Software Technology, Nanjing University}
	\city{Nanjing}
	\country{China}
}
\email{zychen@nju.edu.cn}

\begin{abstract}
	Deep neural networks (DNN) have been deployed in many software systems to assist in various classification tasks.
In company with the fantastic effectiveness in classification, DNNs could also exhibit incorrect behaviors and result in accidents and losses.
Therefore, testing techniques that can detect incorrect DNN behaviors and improve DNN quality are extremely necessary and critical.
However, the testing oracle, which defines the correct output for a given input, is often not available in the automated testing.
To obtain the oracle information, the testing tasks of DNN-based systems usually require expensive human efforts to label the testing data, 
which significantly slows down the process of quality assurance.

To mitigate this problem, we propose DeepGini, a test prioritization technique designed based on a statistical perspective of DNN.
Such a statistical perspective allows us to reduce the problem of measuring misclassification probability to the problem of measuring set impurity,
which allows us to quickly identify possibly-misclassified tests.
To evaluate, we conduct an extensive empirical study on popular datasets and prevalent DNN models. 
The experimental results demonstrate that DeepGini outperforms existing coverage-based techniques in prioritizing tests regarding both effectiveness and efficiency.
Meanwhile, we observe that the tests prioritized at the front by DeepGini are more effective in improving the DNN quality in comparison with the coverage-based techniques.

\end{abstract}

% 	<ccs2012>
% 	<concept>
% 	<concept_id>10011007.10011074.10011099.10011102.10011103</concept_id>
% 	<concept_desc>Software and its engineering~Software testing and debugging</concept_desc>
% 	<concept_significance>500</concept_significance>
% 	</concept>
% 	</ccs2012>

% \ccsdesc[500]{Software and its engineering~Software testing and debugging}

%    <concept>
%        <concept_id>10011007.10011074.10011099.10011102.10011103</concept_id>
%        <concept_desc>Software and its engineering~Software testing and debugging</concept_desc>
%        <concept_significance>300</concept_significance>
%        </concept>
%  </ccs2012>



\maketitle

\section{Introduction}
\label{sec:intro}

We are entering the era of deep learning, which has been widely adopted in many areas.
Famous applications of deep learning include image classification \cite{he2016deep}, autonomous driving \cite{bojarski2016end},
speech recognition \cite{xiong2016achieving}, playing games \cite{silver2016mastering}, and so on.
Although for the well-defined tasks, such as in the case of Go \cite{silver2016mastering}, deep learning has achieved or even surpassed the human-level capability, it still has many issues on reliability and quality. 
These issues could cause significant losses such as in the accidents caused by the self-driving car of Google and Tesla~\cite{GoogleAccident,TeslaAccident,TeslaAccident2}.

Testing is considered to be the common practice for software quality assurance.
However, testing on DNN-based software is significantly different from conventional software because, while conventional software depends on programmers to manually build up the business logic, DNNs are constructed based on a data-driven programming paradigm. 
Thus, sufficient test data, with oracle information, is critical for detecting misbehaviors of DNN-based software.
Unfortunately, like the testing techniques for conventional software, DNN testing also faces the problem that automated testing oracle is often unavailable.
For example, it costs more than 49,000 workers from 167 countries for about 9 years to label the data in ImageNet~\cite{deng2009imagenet}, one of the largest visual recognition datasets containing millions of images in more than 20,000 categories.

Specially, in the context of testing DNN-based systems, 
software testers often focus on the tests that can cause the system to behave incorrectly, because diagnosing these tests can provide insights into various problems in the program.
This fact naturally motivates us to propose a technique to prioritize tests so that
fault-inducing tests can be labeled and analyzed before the other tests.
In this manner, we can obtain maximum benefit from human efforts, even the labeling process is prematurely halted at some arbitrary point due to resource limit.

In the past decades, many test prioritization techniques have been proposed for conventional software systems~\cite{rothermel2001prioritizing,di2013coverage,yoo2012regression}. 
In these techniques, code coverage is employed as the metric to guide the prioritization procedure.
Two main coverage-based techniques are known as coverage-total and coverage-additional test prioritization~\cite{yoo2012regression}.
A coverage-total method prioritizes tests based on their individual total coverage rate. 
That is, we prefer a test to the other one if it covers more program elements. 
A coverage-additional method differs from the coverage-total method in that,
it prefers a test if it can cover more program elements that have not been covered by previous tests.

Unfortunately, for DNN-based systems, although several neuron-coverage criteria have been proposed by software engineering researchers~\cite{pei2017deepxplore,ma2018deepgauge},
the aforementioned coverage-based methods are not effective as expected, due to some new challenges:
\begin{itemize}
	\item First, these criteria are proposed to measure testing adequacy.
	It is often not clear how to improve the DNN quality after testing a DNN with these coverage criteria.
	
	\item Second,
	some coverage criteria cannot distinguish the fault detection capability of different tests. 
	Thus, we cannot prioritize them effectively using the coverage-total prioritization method.
	For example, given a DNN, every test of the DNN has the same top-$k$ neuron coverage rate~\cite{ma2018deepgauge}.
	As a result, the coverage-total method becomes meaningless using this coverage criterion.
	
	\item Third, for most of existing neuron coverage criteria, only a few tests in a test set can achieve the maximum coverage rate of the set.
	For example, using the top-$k$ neuron coverage~\cite{ma2018deepgauge}, we only need about 1\% tests in a test set to achieve the maximum coverage rate of the test set.
	In this case, the coverage-additional method becomes useless because it stops working after prioritizing the first 1\% tests.
	
	\item Fourth, coverage-additional method usually takes very high time complexity $O(mn^2)$, where $m$ is the number of elements, e.g., neurons, to cover and $n$ is the number of tests. Since $n$ and $m$ are usually very large for a DNN, this method is not scalable.
\end{itemize}

To overcome the aforementioned problems and effectively improve the DNN quality, 
in this paper, we propose a test prioritization technique called DeepGini, especially for image-classification DNNs.
DeepGini does not prioritize tests as conventional coverage-based approaches but is based on
a statistical perspective of DNN. Such a statistical perspective allows us to reduce the problem of measuring misclassification probability to the problem of measuring set impurity~\cite{quinlan1986induction}.
Intuitively,
a test is likely to be misclassified by a DNN if the DNN outputs similar probabilities for each class.
Thus, this metric yields the maximum value when DNN outputs the same probability for each class.
For example, if a DNN outputs a vector $\langle 0.5, 0.5 \rangle$,
it means that the DNN is not confident about its classification because
the test has the same probability (i.e., $0.5$) to be classified into the two classes.
In this case, the DNN is more likely to make mistakes.
In contrast, if the DNN outputs $\langle 0.9, 0.1 \rangle$,
it implies that the DNN is confident that the test should be classified into the first class.
Compared to the coverage-based approaches,
DeepGini has the following advantages:
\begin{itemize}
	\item Tests are more distinguishable using our metric than existing neuron coverage criteria.
	This is because it is not likely that different tests have the same output vector
	but tests usually have the same coverage rate as discussed above.
	
	\item It is not necessary for us to record a great deal of intermediate information to compute coverage rate. We prioritize tests only based on the output vector of a DNN. 
	Since it is not necessary for us to understand the internal structure of a DNN,
	our approach is much easier to use. Meanwhile, it is also more secure because we do not need to look into a DNN and, thus, sensitive information in a DNN is protected.
	
	\item The time complexity of DeepGini is the same with the coverage-total approach but much less than the coverage-additional approach. Thus, our approach is as scalable as coverage-total approaches but much more scalable than coverage-additional approaches.
	
	\item Tests prioritized at the front by DeepGini are more effective to improve the DNN quality than those prioritized at the back and those prioritized at the front but by coverage-based prioritization techniques.
\end{itemize}

We notice that our approach requires to run all tests to obtain the output vectors so that the likelihood of misclassification can be calculated.
However, we argue that this is not a significant weakness.
First, this issue is shared with all coverage-based prioritization methods
as they also need to run tests to obtain the coverage rates.
Second, the time cost to run a DNN is not time-consuming like training the DNN.
Compared to the expensive cost of manually labeling all tests in a messy order,
the time cost is completely negligible.

We compare DeepGini with coverage-based methods using existing neuron coverage criteria.
The effectiveness of our approach is evaluated from two aspects.
First,
we compute the value of Average Percentage of Fault-Detection (APFD)~\cite{yoo2012regression},
which is a standard method of evaluating prioritization techniques.
Second,
we evaluate if our technique can improve the quality of DNN.
To this end, we add the tests prioritized at the front to the training set and compare the accuracy of the re-trained DNN to the original one.
The experimental results demonstrate that DeepGini is close to the optimal solution in terms of the value of APFD,
and DeepGini is also more effective to improve the DNN quality.
In summary, we make the following contributions in this paper:
\begin{itemize}
	\item We propose an effective and efficient approach, DeepGini, to prioritizing DNN tests. 
	\item We demonstrate that tests prioritized at the front by DeepGini are effective to improve the DNN quality.
	\item We show the weaknesses of neuron coverage criteria in test prioritization and DNN enhancement.
%	\item We extensively evaluate our method and demonstrate that it is much more effective than coverage-based methods.
\end{itemize}

The remainder of the paper is organized as follows. 
Section \ref{sec:bkg} introduces the background knowledge of deep learning and test prioritization.
Section \ref{sec:approach} presents our approach to prioritizing tests and its application to improving the DNN quality. 
Section \ref{sec:expdesign} takes us to the empirical study, in which we introduce the settings of the evaluation.
Section~\ref{sec:resultanalysis} discusses the experimental results, which
demonstrate the effectiveness and efficiency of our approach.
Section \ref{sec:relwork} surveys the related work and Section \ref{sec:conc} concludes this paper.

\section{Background}
\label{sec:bkg}

In this section, we introduce the basic knowledge of DNN and conventional test prioritization techniques.

\subsection{Deep Neural Networks}
\label{subsec:dnn}

Classification deep neural networks (DNN) are the core of many deep learning systems.
As shown in Figure \ref{fig:dnn}, a DNN consists of multiple layers, i.e., an input layer, an output layer, and one or more hidden layers.
Each layer is made up of a series of neurons. The neurons from different layers are interconnected by weighted edges.
Each neuron is a computing unit that applies an activation function on its inputs and the weights of the incoming edges.
The computed result is passed to the next layer through the edges.
The weights of the edges are not specified directly by software developers, but automatically learned by a training process with a large set of labeled training data.
After training, a DNN then can be used to automatically classify an input object, e.g., an image with an animal, into its corresponding class, e.g., the animal species.

\begin{figure}[t]
	\centering
	\includegraphics[width=\linewidth]{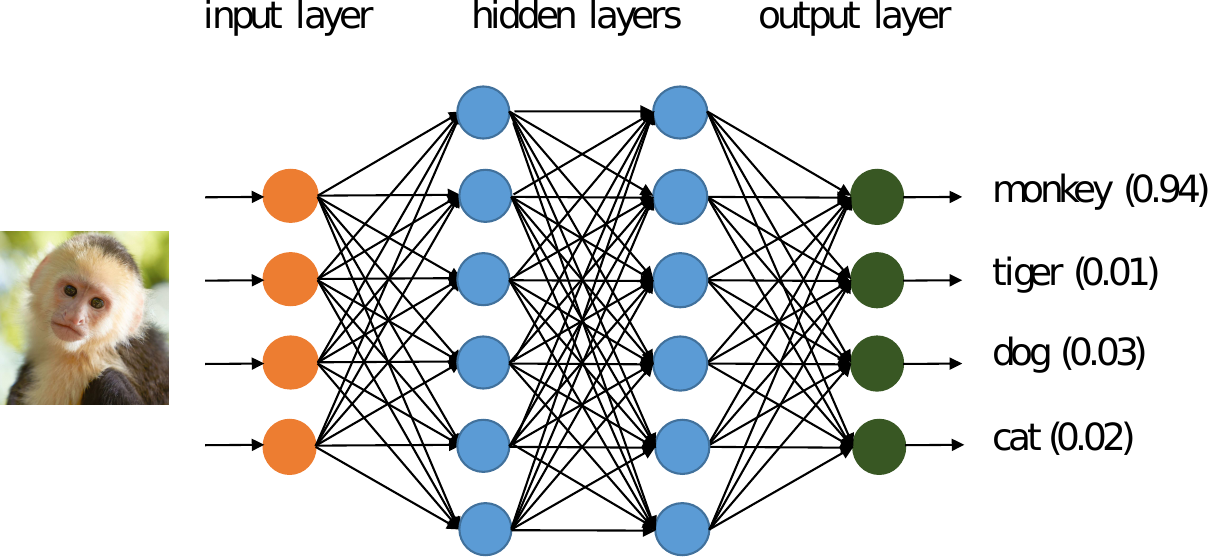}
	\caption{An example to illustrate the DNN structure.}
	\label{fig:dnn}
\end{figure}

Suppose we have a DNN that classifies objects into $N$ classes.
Given an input, the DNN will output a vector of $N$ values, e.g., $\langle v_1, v_2, \cdots, v_N\rangle$, each of which represents how much the system thinks the input corresponds to each class.
Using a softmax function~\cite{denker1991transforming}, it is easy to normalize this vector to $\langle p_1, p_2, \cdots, p_N\rangle$ where
$\Sigma_{i=1}^N p_i = 1$ and $p_i$ indicates the probability that an input belongs to the \textit{i}th class.
From now on, with no loss of generality, we assume that the output vector of a DNN is a vector of probabilities as described above.

\subsection{Neuron Coverage Criteria}
\label{subsec:coverage_dnn}

To enhance the quality and robustness of DNNs, in the past decade, software engineering researchers have proposed a series of neuron coverage criteria specifically for DNN testing~\cite{pei2017deepxplore,ma2018deepgauge,wicker2018feature,kim2019guiding}.
In this section, we survey the related papers published on peer reviewed venues as follows.

\textbf{Neuron Activation Coverage (NAC($k$))} \cite{pei2017deepxplore}.
NAC($k$) is proposed based on the assumption that
higher activation coverage implies that more states of a DNN could be explored.
%Thus we have more opportunities to find defects.
The parameter $k$ of this coverage criterion is defined by users 
and specifies how a neuron in a DNN can be counted as covered.
That is, if the output of a neuron is larger than $k$, then this neuron will be counted as covered.
The rate of NAC($k$) for a test is defined as the ratio of the number of covered neurons to the total number of neurons.

\textbf{$k$-Multisection Neuron Coverage (KMNC($k$))} \cite{ma2018deepgauge}.
Suppose that the output of a neuron $o$ is located in an interval $[\textit{low}_o, \textit{high}_o]$,
where $\textit{low}_o$ and $\textit{high}_o$ are recorded in the training process.
To use this coverage criterion,
the interval $[\textit{low}_o, \textit{high}_o]$ is divided into $k$ equal sections,
and the goal is to cover all the sections.
We say a section is covered by a test if and only if the neuron output is located in the section when the DNN is run against the test.
The rate of KMNC($k$) for a test is defined as the ratio of the number of covered sections to the total number of sections.
Here, the total number of sections is equal to $k$ times the total number of neurons.

In most cases, a single test must cover a section in $[\textit{low}_o, \textit{high}_o]$ of each neuron.
Only a tiny number of tests do not cover a section in the interval, but cover the boundaries, i.e., $(-\infty, \textit{low}_o]$ and $[\textit{high}_o, +\infty)$.
Thus, almost all single tests have the same coverage rate of KMNC($k$).

\textbf{Neuron Boundary Coverage (NBC($k$))}~\cite{ma2018deepgauge}.
Different from KMNC($k$),
NBC($k$) does not aim to cover all sections in $[\textit{low}_o, \textit{high}_o]$.
Instead, it targets to cover the boundaries, i.e., $(-\infty, \textit{low}_o]$ and $[\textit{high}_o, +\infty)$.
Using this coverage criterion, we can expect to cover more corner cases.
In practice, 
it is not necessary to directly use $\textit{low}_o$ and $\textit{high}_o$ as the boundaries.
Instead, $\textit{low}_o - k\sigma$ and $\textit{high}_o + k\sigma$ can be used.
Here, $\sigma$ is the standard deviation of the outputs of a neuron recorded in the training process.
$k$ is a user-defined parameter.
The rate of NBC($k$) for a test is defined as the ratio of the number of covered boundaries to the total number of boundaries.
Since each neuron has one upper bound and one lower bound,
the total number of boundaries is twice the number of neurons.

\textbf{Strong Neuron Activation Coverage (SNAC($k$))} \cite{ma2018deepgauge}.
SNAC($k$) can be regarded as a special case of NBC($k$) as it only takes upper boundary into consideration.
Thus, it is defined as the ratio of the number of covered upper boundaries
to the total number of upper boundaries, in which the latter is actually equal to the number of neurons in a DNN.

\textbf{Top-$k$ Neuron Coverage (TKNC($k$))} \cite{ma2018deepgauge}.
TKNC($k$) measures how many neurons have once been the most active $k$ neurons on
each layer. It is defined as the ratio of the total number of top-$k$
neurons on each layer to the total number of neurons in a DNN.
We say a neuron is covered by a test 
if and only if when the DNN is run against the test, the output of the neuron is larger than or equal to the $k$th highest value in the layer of the neuron.

It is noteworthy that, according to this definition, this metric only can be used to compare two test sets with more than one test.
For a single test, it always covers $k$ neurons in each layer of a DNN.
Thus, TKNC($k$) is always the same for two single tests and, thus, cannot 
distinguish them.

\textbf{Likelihood- and Distance-based Surprise Coverage (LSC($k$) and DSC($k$))} \cite{kim2019guiding}.
Surprise coverage relies on the concept of surprise adequacy $SA(x)$,
which measures the dissimilarity between a test $x$ and the training data set.
The parameter $k$ here is a pair $(n, u)$.
Given an upper bound $u$ and buckets $B=\{b_1, b_2, \cdots, b_n\}$ that divides $(0, u]$ into $n$ $SA$ segments,
the surprise coverage rate of a set $X$ of tests is defined as
$$
SC(X) = \frac{|\{ b_i | \exists x\in X: SA(x)\in (u*\frac{i-1}{n}, u*\frac{i}{n}] \}|}{n}
$$
\\

LSC and DSC, two special types of surprise coverage, rely on 
likelihood-based surprise adequacy (LSA) and distance-based surprise adequacy (DSA), respectively.
LSA uses kernel density estimation~\cite{matt1994kernel} to estimate the surprise adequacy
while DSA uses Euclidean distance.
The details on the computation of DSA and LSA are omitted
and can be found in the original paper~\cite{kim2019guiding}.

\subsection{Coverage-Based Test Prioritization}
\label{subsec:coverage-tcp}

In conventional software testing, test prioritization (a.k.a., test case prioritization) is actually a classic problem defined by Rothermel et al. \cite{rothermel2001prioritizing} as following:

\textbf{Test Prioritization.} Given a test set $T$, the set $PT$ of the permutations of $T$, and a function $f$ from $PT$ to the real numbers, the test prioritization problem is to find $T'\in PT$ such that 
$$
\forall T''\in PT\setminus\{T'\}: f(T') \ge f(T'').
$$
Here, $f(T'\in PT)$ yields an award value for a permutation.

In the past decades, many test prioritization techniques have been proposed for conventional software.
Most of these techniques are based on various code coverage information and follow the basic assumption that early maximization of coverage would lead to early detection of faults~\cite{di2013coverage}. 
Two main coverage-based techniques are known as the coverage-total prioritization and the coverage-additional prioritization~\cite{yoo2012regression}.

\textbf{Coverage-Total Method (CTM).}
A CTM is an implementation of the ``next best'' strategy.
It always selects the test with the highest coverage rate,
followed by the test with the second-highest coverage rate, and so on.
For tests with the same coverage rate, the method will prioritize them randomly.
For the example in Table \ref{tab:tcpex},
both $A, B, C, D$ and $A, B, D, C$ are valid results of CTM.

CTM is attractive because it is relatively efficient and easy to implement.
Given a set consisting of $n$ tests with their coverage rates,
CTM only needs to sort these tests according to their coverage rates.
Typically, using a quick sort algorithm, it only takes $O(n\log n)$ time \cite{cormen2009introduction}.

\textbf{Coverage-Additional Method (CAM).}
CAM differs from CTM in that it selects the next test according to the feedback from previous selections.
It iteratively selects a test that can cover more uncovered code structures.
In this manner, we can expect that we can achieve the maximum coverage rate of a test set as soon as possible.
After the maximum coverage rate is achieved,
we can use CTM to prioritize the remaining unprioritized tests.
For the example in Table \ref{tab:tcpex},
$A, D, C, B$ is the only valid result of CAM.

Given a program with $m$ elements to cover and a set of $n$ tests,
every time we select a test, it will take $O(mn)$ time to re-adjust the coverage information of the remaining tests.
This process will be performed $O(n)$ times. Thus, the total time cost is $O(mn^2)$.
According to the time complexity, it is easy to find that CAM is less scalable compared to CTM,
especially when $n$ and $m$ are very large.

\begin{table}[t]
	\centering
	\caption{An example to illustrate coverage-based test prioritization. `X' means a statement is covered by a test.}
	\begin{tabular}{c|cccccccc}
		\toprule
		\multirow{2}{*}{Test} & \multicolumn{8}{c}{Program Statement} \\
		& 1  & 2  & 3  & 4  & 5  & 6  & 7  & 8  \\
		\midrule
		$A$                     & X  & X  & X  &    &    & X  & X  & X  \\
		$B$                     & X  & X  & X  &    &    &    & X  & X  \\
		$C$                     & X  & X  & X  & X  &    &    &    &    \\
		$D$                     &    &    &    &    & X  & X  & X  & X \\
		\bottomrule
	\end{tabular}
	\label{tab:tcpex}
\end{table}

\section{Approach}
\label{sec:approach}

Owing to the oracle problem discussed before,
test prioritization can help label and analyze as many misclassified tests as possible in a limited time.
However, due to the problems we argued in Section \ref{sec:intro} and as we will show in our evaluation,
coverage-based test prioritization becomes ineffective in the context of DNN testing.
Therefore, we propose a test prioritization method that is not based on neuron coverage criteria,
but based on a statistical view of DNN as discussed in Section \ref{subsec:view}.
Such a statistical view inspires us to propose a method, called DeepGini, 
to prioritize tests of a DNN, which is presented in Section \ref{subsec:tp}.
Section~\ref{subsec:app} discusses how to improve DNN quality with DeepGini.

\subsection{A Statistical View of DNN}
\label{subsec:view}

DNNs are specially good at classifying high-dimensional objects. 
If we regard each output class of a DNN as a kind of feature of the input object,
the computation (or classification) process of a DNN actually maps the original high-dimensional data to only a few kinds of features.
As an example, suppose the input of a DNN is a 28x28 image with three channels (i.e., RGB channels). 
Then the original dimension of the image is $3^{28\times 28}$. 
In Figure \ref{fig:dnn}, the DNN maps the high-dimension object to a bag (or multi-set)\footnote{A multi-set or a bag is a special kind of set that allows duplicate elements.} 
$B$ of features, 
in which 94\% are features of monkey, 1\% are features of tiger, 3\% are features of dog, and 2\% are features of cat.
Since most elements in $B$ are features of monkey, we classify the input object into the monkey class.

Generally, if the feature bag has the highest purity, i.e., contains only one kind of features (e.g., 100\% elements in $B$ are features of monkey), 
then there will be no other features confusing our classification and it is more likely that a test input is correctly classified.
As an example, in Figure~\ref{fig:purity},
Bag 2 has higher purity than Bag 1.
Intuitively, this is because almost all elements in Bag 2 are triangles while Bag 1 has the same number of triangles and circles.

Statistically, if a bag has higher purity, the results of two random samplings in the bag have higher probability to be the same.
In contrast, if a bag has lower purity, the results of two random samplings in the bag are more likely to be different.
For the example in Figure~\ref{fig:purity},
using sampling with replacement,\footnote{In sampling with replacement, 
	after we sample a feature from the feature bag, the feature is put back to the bag so that we have the same probability to get the feature next time.}
the probability that two random samplings have the same shape is $0.5^2 + 0.5^2 = 0.5$ and $0.1^2 + 0.9^2 = 0.82$ for Bag 1 and Bag 2, respectively.

Formally, assuming the feature distribution in the feature bag output by a DNN is $\langle p_1, p_2, \cdots, p_N\rangle$,
we can compute the probability that two random samplings have different results as $1 - \Sigma_{i=1}^N p_i^2$.
The lower the probability, the higher the purity and, thus, the more likely a test input of a DNN is correctly classified.

\begin{figure}[t]
	\centering
	\includegraphics[width=\linewidth]{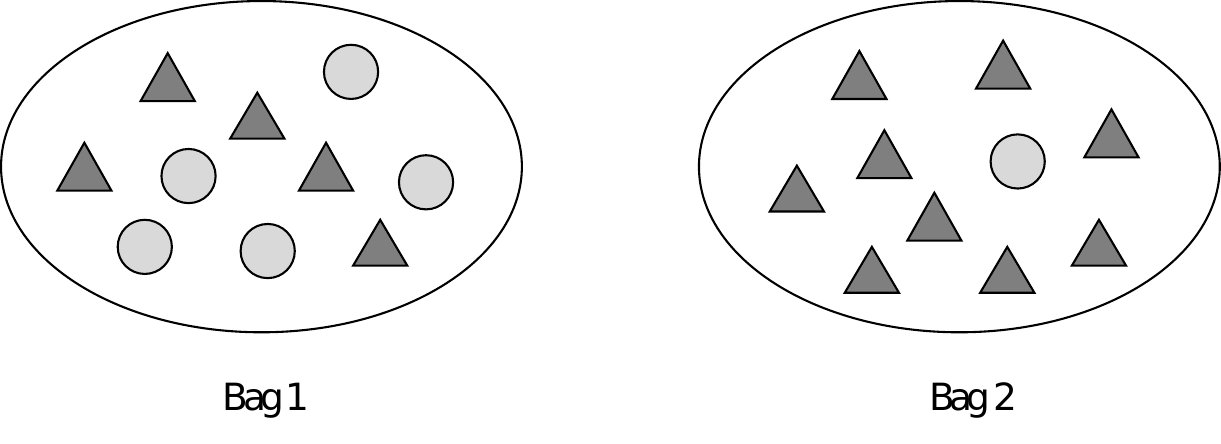}
	\caption{Bag 2 has higher purity than Bag 1. Bag 1 has 50\% triangles and 50\% circles. Bag 2 has 90\% triangles and 10\% circles.}
	\label{fig:purity}
\end{figure}

On the statistical view, we can observe that the problem of measuring the likelihood of misclassification actually has been reduced to the problem of measuring the purity of a bag.
In fact, such a reduction follows the very spirit of the measurement of Gini impurity~\cite{quinlan1986induction},
which inspires us to propose DeepGini for measuring the likelihood of misclassification.

\subsection{DeepGini: Prioritizing Tests of a DNN}\label{subsec:tp}

Formally, the metric we use to measure the likelihood of misclassification is defined as below.

\begin{definition}
	Given a test $t$ and a DNN that outputs $\langle p_{t,1}, p_{t,2}, \cdots,$ $p_{t,N}\rangle$ $(\Sigma_{i=1}^N p_{t, i} = 1)$,
	we define $\xi(t)$ to measure the likelihood of $t$ being misclassified:
	$$
	\xi(t) = 1 - \Sigma_{i=1}^N p_{t, i}^2
	$$
\end{definition}

%A test is likely to be misclassified if the DNN outputs similar probabilities for each class.
%For example, if a test has a probability of 0.5 belonging to Class A and Class B, respectively,
%then it will be difficult for us to determine the class of the test.
%Therefore, the metric to measure the likelihood of a test being misclassified
%actually measures the degree to which the probability of a test is spread out over different classes.
%The more classes that a test has the same probability of belong to,
%the larger the value of the metric.
%The metric has its unique extremum if and only if the test has the same probability of belonging to all classes.
%To achieve the above properties, we define the metric as below.

In the definition, $p_{t, i}$ is the probability that the test $t$ belongs to the class $i$.
Figure \ref{fig:gini} illustrates the distribution of $\xi$ when the DNN performs a binary classification.
The distribution illustrates that when DNN outputs the same probability for the two classes,
$\xi$ has the maximum value, indicating that we have high probability to incorrectly classify the input test.
This result follows our intuition that a test is likely to be misclassified if the DNN outputs similar probabilities for each class,
and the rationality of the result has been explained in the previous subsection.
The following theorem demonstrates that even though a DNN classifies input tests into more than two classes,
$\xi$ has a similar distribution as that in Figure~\ref{fig:gini}.

\begin{theorem}
	$\xi(t)$ has the unique maximum if and only if $\forall 1\le i, j\le N: p_{t, i} = p_{t, j}$.
\end{theorem}
\begin{proof}
	According to Lagrangian multiplier method \cite{rockafellar1993lagrange}, let
	$$
	L(p_{t, i}, \lambda) = \xi(t) + \lambda\times(\Sigma_{i=1}^N p_{t, i} - 1)
	$$
	$\forall p_{t, i}$, let
	$$
	\frac{\partial L}{\partial p_{t, 1}} = -2 p_{t, 1} + \lambda =0
	$$
	$$
	\frac{\partial L}{\partial p_{t, 2}} = -2 p_{t, 2} + \lambda =0
	$$
	$$
	\vdots
	$$
	$$
	\frac{\partial L}{\partial p_{t, N}} = -2 p_{t, N} + \lambda =0
	$$
	
	If we calculate the difference of any two above equations
(e.g. the $i$th and $j$th equation), we will have
$$
   2 p_{t, i} - 2 p_{t,j} = 0 \Rightarrow p_{t,  i} = p_{t, j}
$$

Hence, when $p_{t, 1} =p_{t, 2} = \cdots = p_{t, N} = 1/N$, $\xi(t)$ has the unique extremum.

At the point $(p_{t, 1}, p_{t, 2}, \cdots, p_{t, N})$, the Hessian matrix \cite{binmore2002calculus} of $\xi$ is
$$
\begin{bmatrix}
-2 & 0 & \dots  & 0 \\
0 & -2 & \dots  & 0 \\
\vdots & \vdots & \ddots & \vdots \\
0 & 0 & \dots  & -2
\end{bmatrix}
$$
which is a negative definite matrix. This implies that the unique extremum must be the unique maximum \cite{binmore2002calculus}.

%Next, we use a concrete example to show that the unique extremum must be the unique minimum.
%Suppose $N=2$ and we have two output vectors $\langle 0.5, 0.5\rangle$ and $\langle 0.4, 0.6\rangle$.
%Obviously, $0.5^2 + 0.5^2 = 0.5 < 0.52 = 0.4^2 + 0.6^2$. 
%Thus, the unique extremum must be the unique minimum.
\end{proof}

% \bigskip

\begin{figure}
	\centering
	\includegraphics[width=\linewidth]{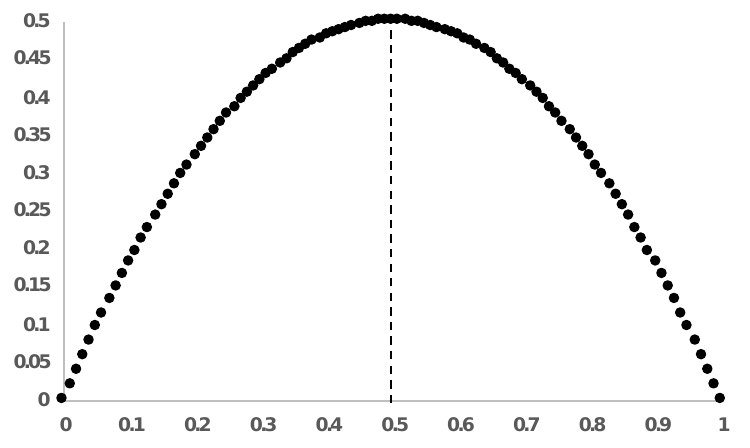}
	\caption{Distribution of $\xi$ for 2-class problem. X-Axis: the probability that a test input belongs to one of the two classes. Y-Axis: the value of $\xi$.}
	\label{fig:gini}
\end{figure}

We notice that many other metrics such as information entropy~\cite{shannon1948mathematical} also have the above property and is almost equivalent to $\xi$~\cite{raileanu2004theoretical}.
The difference is that it may require a non-statistical view, e.g., the perspective of information theory, to explain the rationality.
In addition, we believe that the simplest is the best: the complexity of computing quadratic sum is much easier than that of computing entropy-like metrics because they require logarithmic computation. 

% \bigskip

%\subsection{Test Prioritization}\label{subsec:tp}

According to the above discussion,
$\xi(t_1) > \xi(t_2)$ implies that $t_1$ is more likely to be misclassified.
Hence, to prioritize $n$ tests in a set,
we need to run the tests to collect the outputs,
and then sort these tests $t_i$ according to the value of $\xi$.

We argue that the time cost of running the tests is negligible.
First, the time cost to run a DNN is not time-consuming like training the DNN.
Compared to the expensive cost of manually labeling all tests in a messy order,
the time cost is completely negligible.
Second, this issue is shared with all neuron-coverage-based test prioritization methods
as they also need to to run tests to obtain the coverage rates.

\begin{example}
	Assume that we have four tests $A, B, C$, and $D$ as well as 
	a DNN tries to classify them into three classes. Table \ref{tab:xi}
	shows their output vectors and the values of $\xi$.	
	According to the values of $\xi$, we can prioritize the tests as $D, A, C$, and $B$.
	$D$ has the highest probability to be misclassified because the DNN outputs the most similar probabilities for each of the three classes.
	In comparison, for $B$ and $C$, the DNN is more confident about their classes as $B$ has the probability of 0.8 to be classified into the third class
	and $C$ has the probability of 0.6 to be classified into the first class.
\end{example}

\begin{table}[t]
	\centering
	\caption{An example to show how DeepGini prioritizes tests.}
	\begin{tabular}{c|cc}
		\toprule
		Test & Output of DNN & $\xi$ \\
		\midrule
		$A$                     & $\langle 0.3, 0.5, 0.2\rangle$ & 0.62  \\
		$B$                     & $\langle 0.1, 0.1, 0.8 \rangle$  & 0.34   \\
		$C$                     & $\langle 0.6, 0.3, 0.1\rangle$ & 0.54    \\
		$D$                     & $\langle 0.4, 0.4, 0.2\rangle$   & 0.64 \\
		\bottomrule
	\end{tabular}
	\label{tab:xi}
\end{table}

Typically, in our prioritization method, we can simply use a quick sort algorithm to sort tests.
This algorithm takes $O(n\log n)$ time complexity.
Compared to CTM and CAM, our approach has following merits:
\begin{itemize}
	\item The time complexity of our approach is the same with CTM and is much lower than CAM ($O(mn^2)$). Thus, our approach is as scalable as CTM and much more scalable than CAM.
	\item Different from CTM and CAM, we only need to record output vectors while CTM and CAM require us to profile the whole DNN to record coverage information. Thus, our approach has less interference with the DNN.
\end{itemize}

\subsection{Enhancing DNN with DeepGini}
\label{subsec:app}

Generally, we can add more tests to the training set and retrain the DNN to enhance its robustness.
In face of a large number of unlabeled tests and a limited time budget,
we cannot label all tests and use them to retrain the DNN.
DeepGini allows us to find and label as many misclassified tests as possible in a limited time budget.
We observe that the tests prioritized by DeepGini at the front are more effective to improve DNN quality than the tests prioritized at the back.
Meanwhile,
our empirical study shows that tests prioritized by DeepGini at the front are more effective to improve DNN quality than the tests prioritized at the front but by coverage-based prioritization techniques.

The principle behind the effectiveness of DeepGini actually follows the theory of \textit{active machine learning},
which prefers the tests near the decision boundary (i.e., tests that the DNN is least certain how to label or, equivalently, tests that have the highest value of $\xi(t)$) to actively enhance a deep learning system.
We omit the details of active learning here because it is not our contribution. Interested readers can refer to the literature~\cite{settles2009active} for more details.

To sum up, DeepGini provides not only a test prioritization method but also a technique to enhance the robustness of DNN in a limited time budget.

\begin{table*}[t]
	\centering
	\caption{Datasets and DNN models.}
	\label{tab:dataset}
	\def\arraystretch{0.9}
	\begin{tabular}{l|cccc|cc}
		\toprule
		\textbf{Dataset} & \textbf{Description}      & \textbf{DNN Model} & \textbf{\#Neurons} & \textbf{\#Layers} & \# Original Tests & \# Adversarial Tests \\
		\midrule
		\multirow{2}{*}{MNIST}   & \multirow{2}{*}{Digits 0$\sim$9}  & LeNet-1   & 42    & 5     & \multirow{2}{*}{10,000} & 39,854     \\
	                           	 &                                   & LeNet-5   & 258   & 7     &                         & 39,705  \\
		\hline
		\multirow{2}{*}{CIFAR-10}& \multirow{2}{*}{Images with 10 classes} & ResNet-20    & 698   & 20  & \multirow{2}{*}{10,000}& 40,000\\
		                         &                                         & VGG-16       & 7242  & 21  &                        & 8,000 \\
		\hline
		\multirow{2}{*}{Fashion} & \multirow{2}{*}{Zalando's article images} & LeNet-1    & 42   & 5   &\multirow{2}{*}{10,000} & 39,992 \\
		                         &                                           & ResNet-20  & 698  & 20  &                        & 39,905 \\
		\hline
		
		\multirow{2}{*}{SVHN}  & \multirow{2}{*}{Street view house numbers} & LeNet-5  & 258  & 7   & \multirow{2}{*}{26,032} & 104,037  \\
		                       &                                            & VGG-16   & 7242 & 21  &                         & 8,000\\
		\bottomrule
	\end{tabular}
\end{table*}

\section{Experiment Design}
\label{sec:expdesign}

In this section, we introduce the experimental setup,
including the datasets and DNN models we used, approaches to generating adversarial tests, the baseline approaches we compared with, and the research questions we study in the experiments.
To conduct the experiments, we implement our approach as well as various neuron-coverage-based test prioritization methods
upon Keras 2.1.3 with TensorFlow 1.5.0.\footnote{\url{https://faroit.github.io/keras-docs/2.1.3/}}$^,$\footnote{\url{https://github.com/tensorflow/tensorflow/releases}} All of our implementation can be accessed via: \url{https://github.com/deepgini/deepgini}.
All experiments were performed on a Ubuntu 16.04.5 LTS server 
with one NVIDIA GTX 1080Ti GPU, two 12-core processors ``Intel(R) Core(TM) i7-6850K CPU @ 3.60GHz'', and 64GB physical memory.

\subsection{Datasets and DNN Models}

As shown in Table \ref{tab:dataset}, for evaluation,
we select four popular publicly-available datasets, i.e., 
MNIST,\footnote{\url{http://yann.lecun.com/exdb/mnist/}} 
CIFAR-10,\footnote{\url{https://www.cs.toronto.edu/~kriz/cifar.html}}  
Fashion,\footnote{\url{https://research.zalando.com/welcome/mission/research-projects/fashion-mnist/}}  
and SVHN.\footnote{\url{http://ufldl.stanford.edu/housenumbers/}}

The MNIST dataset is for handwritten digits recognition, containing 70,000 input data in total, of which 60,000 are training data and 10,000 are test data. 

The CIFAR-10 dataset consists of 60,000 32x32 colour images in 10 classes, with 6,000 images per class. There are 50,000 training images and 10,000 test images. 

Fashion is a dataset of Zalando's article images consisting of a training set of 60,000 examples and a test set of 10,000 examples. Each example is a 28x28 gray-scale image, associated with a label from 10 classes.

SVHN is a real-world image dataset that can be seen as similar in flavor to MNIST (e.g., the images are of small cropped digits), but incorporates an order of magnitude more labeled data (over 600,000 digit images).

To demonstrate the generalizability of our approach, for every data sets, we use two prevalent DNN models in our evaluation.
The size of these DNN models ranges from tens to thousands of neurons, exhibiting the diversity of the DNN models to some degree.

%On MNIST and CIFAR-10, we use the pre-trained LeNet-5 and ResNet-20
%as the DNN models, respectively.
%For the other two datasets, since we do not find any available pre-trained DNN models, we train the DNN models by ourselves using LeNet-5.

\begin{table}[t]
	\centering
	\caption{Configuration parameters for the coverage criteria.}
	\begin{tabular}{c|c|ccc}
		\toprule
		\textbf{ID} & \textbf{Criteria} & \multicolumn{3}{c}{\textbf{Configuration Parameter} $k$} \\
		\midrule
		1&	NAC($k$)                    & 0  & 0.75  & -   \\
		2&	KMNC($k$)                   & 1,000  & 10,000  & -   \\
		3&	NBC($k$)                   & 0  & 0.5  & 1    \\
		4&	SNAC($k$)                   & 0    &  0.5  & 1    \\
		5&	TKNC($k$)                   &  1  & 2   & 3    \\
		6&	LSC($k$)                    &(1000, 100)   &  -   & - \\
		7&	DSC($k$)                    &(1000, 2)     &  -   & -  \\
		\bottomrule
	\end{tabular}
	\label{tab:param}
\end{table}

\subsection{Adversarial Test Input Generation.}
In addition to prioritizing original tests in the datasets,
we also conduct an experiment to prioritize adversarial tests.
As the previous study~\cite{ma2018deepgauge},
we use four state-of-the-art methods to generate adversarial tests,
including FGSM~\cite{goodfellow2015explaining}, BIM~\cite{kurakin2017adversarial}, JSMA~\cite{papernot2016limitations}, and CW~\cite{carlini2017towards}.
These techniques generate tests through different minor perturbations on a given test input.
Table~\ref{tab:dataset} shows the total number of adversarial tests generated by these methods in 12 hours.

\subsection{Baseline Approaches}
We compare our approach to coverage-based methods that use eleven different neuron coverage criteria as introduced in Section~\ref{sec:bkg}.
Since these neuron coverage criteria contain configurable parameters,
as shown in Table~\ref{tab:param},
we use various parameters as suggested by their original authors.
%Specifically, to reduce the potential bias introduced in our implementation, for all of these baseline approaches, we directly leverage the source code provided by the original authors.

Each comparison experiment is conducted in four modes with regard to two aspects:
(1) using CTM or CAM to prioritize tests; and (2) prioritizing tests in the original datasets or 
prioritizing tests that combine the original tests and the adversarial tests.

\subsection{Research Questions}
\label{sec:rq}

DeepGini is designed for facilitating the testers of DNN-based systems to quickly identify misclassified tests and effectively enhance the robustness.
Based on this goal, we empirically explore the following three research questions (RQ).

\textbf{RQ1. Effectiveness:} Can DeepGini find a better permutation of tests than neuron-coverage-based methods? 

We provide answers to RQ1 by computing the values of Average Percentage of Fault-Detection (APFD) metric~\cite{yoo2012regression}.
Higher APFD values denote faster misclassification-detection rates. 
When plotting the percentage of detected
misclassified tests against the number of prioritized tests, APFD can be calculated as the area below the
plotted line.
It is also noteworthy that although an APFD value ranges from 0 to 1, 
an APFD value not close to 1 does not mean that the prioritization is ineffective.
This is mainly because the theoretically maximal APFD value is usually much smaller than 1 \cite{yoo2012regression}.
Formally, for a permutation of $n$ tests in which there are $k$ tests will be misclassified, 
let $o_i$ be the order of the first test that reveals the $i$th misclassified test.
The APFD value for this permutation can be calculated as following:
$$
APFD = 1- \frac{\Sigma_{i=1}^k o_i}{kn} + \frac{1}{2n}
$$
To be clear, assuming the theoretical minimum and maximum of the APFD value is \textit{min} and \textit{max}, respectively,
we normalize the APFD value from $[\textit{min}, \textit{max}]$ to $[0, 1]$ so that
a prioritization method is better if the APFD value is closer to 1 and is worse if the APFD value is closer to 0.

\textbf{RQ2. Efficiency:} Is DeepGini more efficient than neuron-coverage-based methods?

We provide answers to RQ2 by recording the time cost of prioritization. A prioritization method may be very costly because the number of tests is usually very large for a DNN system.
According to our evaluation, some prioritization methods cannot finish in several hours, which is not practical in an industrial setting.

\textbf{RQ3. Guidance:} Can DeepGini guide the retraining of an DNN to improve its accuracy?

Since the DNNs already have very high accuracy on the original tests ($>90\%$ or even $>95\%$),
we cannot clearly show the accuracy improvement of these DNNs.
Thus, we leverage the adversarial tests to answer RQ3.
For each model, we evenly partition adversarial test set into a testing set $T$ and a validation set $V$ for the following experiment.

After prioritizing the tests in $T$, we add back the first 1\%, 2\%, $\cdots$, 10\% tests into the training set and retrain a new DNN.
We do not add more than 50\% tests to retrain a DNN because we observe that the accuracy of the DNN will not significantly change with more tests.
Using the validation set $V$, we compute the accuracy of the new DNN.
We repeat the experiment using DeepGini and other coverage metrics and compare the accuracy of the retrained DNNs.
According to the experimental results,
we answer RQ3 that
DeepGini can provide guidance for more effective retraining
against the coverage-based methods.

\section{Result Analysis}
% \section{Results, Analysis, and Discussion}
\label{sec:resultanalysis}

% In this section,
% we first discuss the results of test prioritization (RQ1 and RQ2) 
% and then analyze whether our approach can better guide the retraining of DNNs (RQ3).
% Due to the page limit, we cannot present all results in the paper.
% Instead,
% we only present the result of LeNet1 on MNIST dataset as an example.
% All other results are similar and available online: \url{https://github.com/deepgini/deepgini}.

In this section, we present the results of test prioritization (RQ1 and RQ2) and then analyze whether our approach can better guide the retraining of DNNs (RQ3).
Due to the page limit, we cannot present all results in detail.
All other results show similar trends and are available online: \url{https://github.com/deepgini/deepgini}.

\begin{table*}[t]
  \centering
  \caption{Results of Prioritization (MNIST with LeNet5)}
    \def\arraystretch{0.9}
    \begin{tabular}{c|c|c|c|c|c|c|c|c|c|c|c|c|c}
    \toprule
    \multirow{3}{*}{\textbf{Metrics}} & \multirow{3}{*}{\textbf{Param.}} & \multicolumn{6}{c|}{Original Tests}            & \multicolumn{6}{c}{Original Tests + Adv} \\\cline{3-14}
          &       & \multicolumn{2}{c|}{Max Cov.} & \multicolumn{2}{c|}{CTM} & \multicolumn{2}{c|}{CAM} & \multicolumn{2}{c|}{Max Cov.} & \multicolumn{2}{c|}{CTM} & \multicolumn{2}{c}{CAM} \\\cline{3-14}
          &       &   \%    &   \#    & {Time (s)} & {APFD} & {Time (s)} & {APFD} &   \%    &  \#     & {Time (s)} & {APFD} & {Time (s)} & {APFD} \\
    
    \midrule
    
    \multirow{2}[0]{*}{NAC} & 0     & {100} & {1} 	& {2} & {0.638} & 2     & {0.638} & {100} 	& {1} 	& {11} & {0.340} & 11    & {0.340} \\
          					& 0.75  & {84} 	& {22} 	& {2} & {0.385} & 4     & {0.384} & {86} 	& {21} 	& {11} & {0.307} & 16    & {0.307} \\
          
    \midrule
    
    \multirow{3}[0]{*}{NBC} & 0     & {8} 	& {38} 	& {3} & {0.638} & 9    	& {0.638} & {15} 	& {56} 	& {14} & {0.339} & 40    & {0.339} \\
          					& 0.5   & 0.97  & 5   	& {2} & {0.638} & 5     & {0.637} & {3} 	& {11} 	& {13} & {0.400} & 20    & {0.400} \\
          					& 1     & 0.39  & 3   	& {3} & {0.638} & 6     & {0.637} & {2} 	& {7} 	& {13} & {0.400} & 20    & {0.400} \\
          
    \midrule
    
    \multirow{3}[0]{*}{SNAC} 	& 0     & {14} 	& {35} 	& {3} & {0.639} & 9    & {0.639} 	& {22} 	& {48} 	& {13} & {0.340} 	& 31    & {0.340} \\
          						& 0.5   & 2   	& 5   	& {3} & {0.639} & 7    & {0.639}	& {7} 	& {11} 	& {13} & {0.340} 	& 22    & {0.340} \\
          						& 1     & 0.78  & 3	   	& {3} & {0.638} & 8    & {0.638}	& {4} 	& {7} 	& {13} & {0.340} 	& 20    & {0.340} \\
          
    \midrule\midrule
    
    \multirow{3}[0]{*}{TKNC} 	& 1     & {66} & {86} 	& N/A   & N/A   & 11     & {0.023} & {74} & {96} & N/A   & N/A   & 59    & {0.001} \\
          						& 2     & {73} & {67} 	& N/A   & N/A   & 10     & {0.023} & {79} & {70} & N/A   & N/A   & 48    & {0.001} \\
          						& 3     & {76} & {57} 	& N/A   & N/A   & 10     & {0.023} & {81} & {55} & N/A   & N/A   & 43    & {0.001} \\
    
    \midrule
          
    LSC   				& (1000, 100) 	& {22} & {220} 	& N/A 	& N/A 	& 9     & {0.658} 	& {98} & {982} & N/A & N/A & 36    & {0.503} \\
    
    \midrule
    
    DSC   				& (1000, 2) 	& {53} & {531} & N/A 	& N/A & 1177   & {0.658}  	& {97} & {974} & N/A & N/A & 4738  & {0.490} \\
          
    \midrule\midrule
    
    \multirow{2}[0]{*}{KMNC} 	& 1000  & {63} & {8814} & N/A   & N/A   & 34045  & {0.599}	& T/O 	& T/O 	& N/A   & N/A  & T/O & T/O \\
          						& 10000 & T/O  & T/O  	& N/A  & N/A  	& {T/O} & {T/O} 	& T/O  	& T/O  	& N/A  	& N/A  & {T/O} & {T/O} \\
          
    \midrule\midrule

    \textbf{DeepGini} & N/A     & N/A   & N/A   & N/A   & N/A   & \textbf{0.45}  & \textbf{0.984} & N/A   & N/A   & N/A   & N/A   & \textbf{2}     & \textbf{0.991} \\
    \bottomrule
  
    \multicolumn{14}{l}{Max Cov.: The maximum coverage rate of the tests (\%) and the number of tests to achieve the rate (\#).}\\
    \multicolumn{14}{l}{N/A: Not applicable in theory; T/O: Time out, i.e., cannot get result after running for 12 hours.}
    
    \end{tabular}%
  \label{tab:apfd_of_mnist}%
\end{table*}%

% \begin{figure*}[h]
% \centering
% \begin{subfigure}[b]{.495\linewidth}
%   \centering
%   \includegraphics[width=0.495\linewidth]{./figs/apfd-mnist-lenet5-1}
%   \includegraphics[width=0.495\linewidth]{./figs/apfd-mnist-lenet5-2}
%   \caption{}
%   \label{fig:apfd-a}
% \end{subfigure}
% \begin{subfigure}[b]{.495\linewidth}
%   \centering
%   \includegraphics[width=0.495\linewidth]{./figs/apfd-mnist-lenet5-3}
%   \includegraphics[width=0.495\linewidth]{./figs/apfd-mnist-lenet5-4}
%   \caption{}
%   \label{fig:apfd-b}
% \end{subfigure}%
% 	\caption{Test prioritization for MNIST with LeNet5.
% 		X-Axis: the percentage of prioritized tests (sub-figure a), or percentage of both original and adversarial tests (sub-figure b);
% 		Y-Axis: the percentage of detected misclassified tests.}
% \label{fig:apfd}
% \end{figure*}

\begin{figure*}[h]
\centering
\begin{subfigure}[b]{.85\linewidth}
  \centering
  \includegraphics[width=0.495\linewidth]{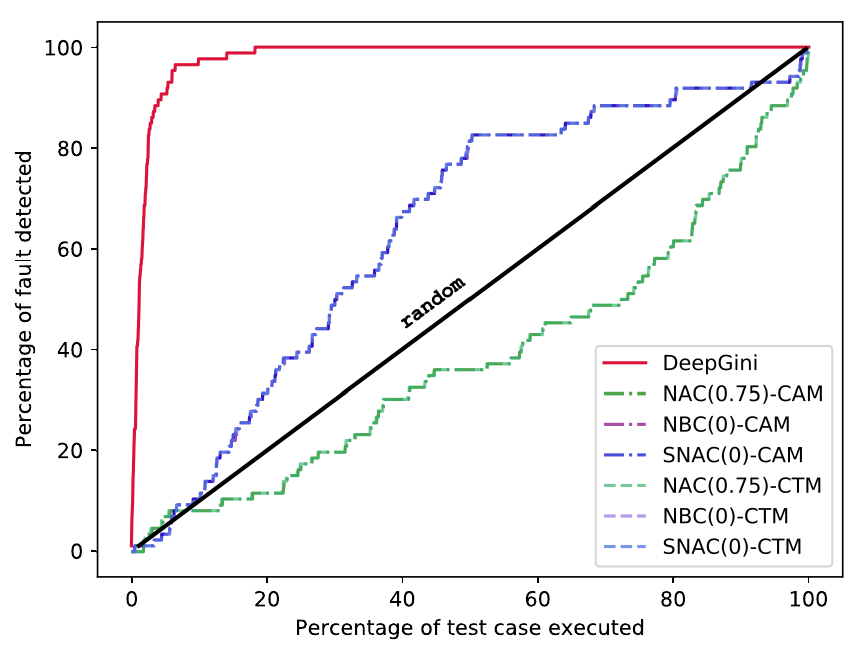}
  \includegraphics[width=0.495\linewidth]{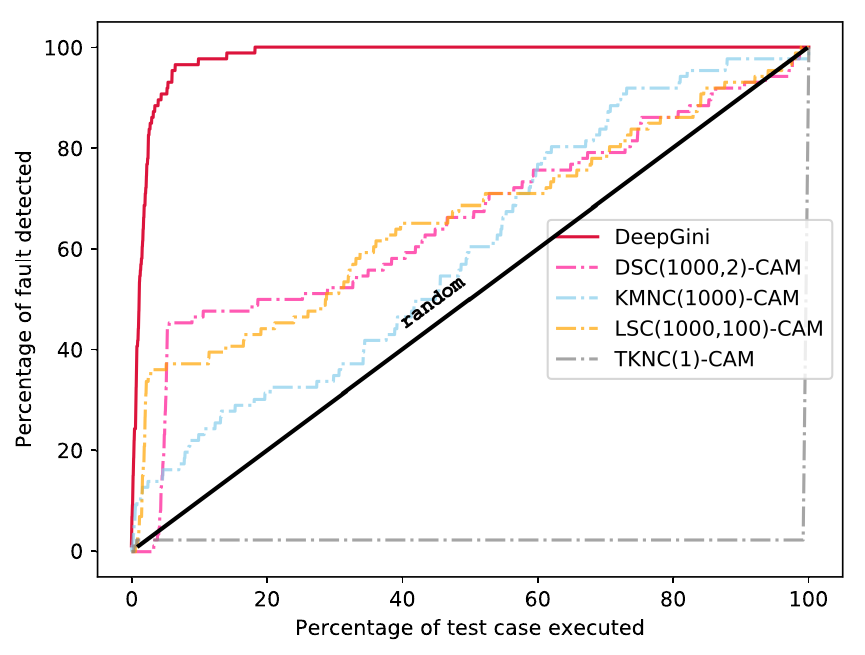}
  \caption{}
  \label{fig:apfd-a}
\end{subfigure}
\begin{subfigure}[b]{.85\linewidth}
  \centering
  \includegraphics[width=0.495\linewidth]{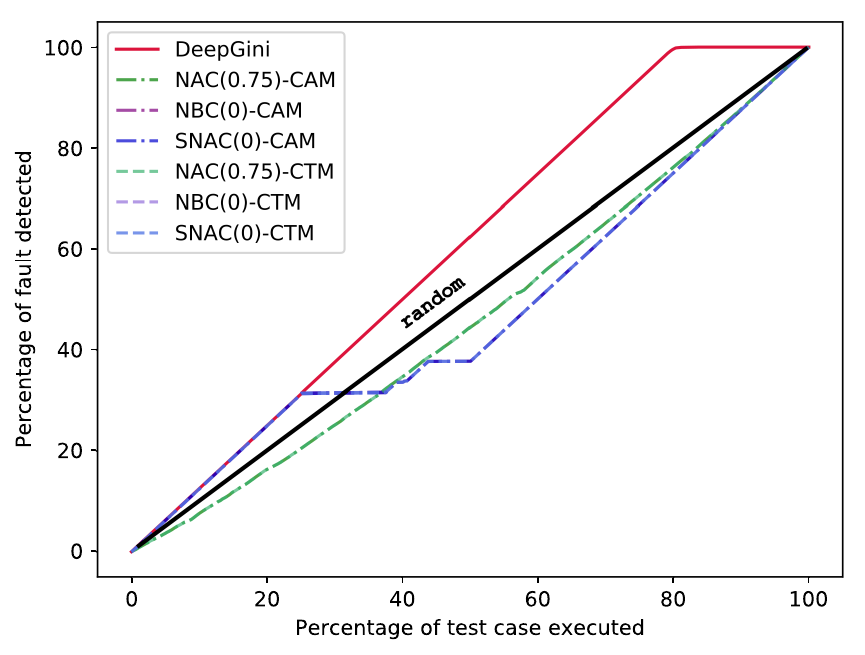}
  \includegraphics[width=0.495\linewidth]{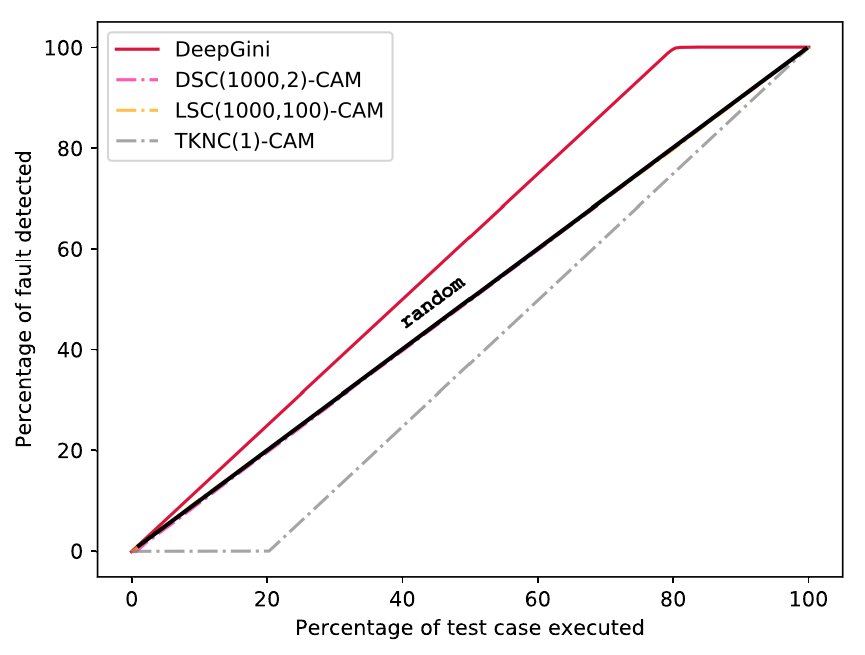}
  \caption{}
  \label{fig:apfd-b}
\end{subfigure}%
  \caption{Test prioritization for MNIST with LeNet5.
    X-Axis: the percentage of prioritized tests (sub-figure a), or percentage of both original and adversarial tests (sub-figure b);
    Y-Axis: the percentage of detected misclassified tests.}
\label{fig:apfd}
\end{figure*}

% \begin{figure*}[h]
% \centering
%   \includegraphics[width=0.4\linewidth]{./figs/apfd-mnist-lenet5-1}
%   \includegraphics[width=0.4\linewidth]{./figs/apfd-mnist-lenet5-2}
%   \caption{Test prioritization for MNIST with LeNet5.
%     X-Axis: the percentage of prioritized tests; 
%     Y-Axis: the percentage of detected misclassified tests.}
%   \label{fig:apfd-a}
% \end{figure*}

% \begin{figure*}[h]
%   \centering
%   \includegraphics[width=0.4\linewidth]{./figs/apfd-mnist-lenet5-3}
%   \includegraphics[width=0.4\linewidth]{./figs/apfd-mnist-lenet5-4}
%   \caption{Test prioritization for MNIST with LeNet5.
%     X-Axis: percentage of both original and adversarial tests;
%     Y-Axis: the percentage of detected misclassified tests.}
%     \label{fig:apfd-b}
% \label{fig:apfd}
% \end{figure*}

\subsection{Effectiveness and Efficiency (RQ1 \& RQ2)}

Table~\ref{tab:apfd_of_mnist} lists the results of the LeNet5 net on the MNIST dataset as an example.
We summarize our findings in Table~\ref{tab:results_of_prioritization} and discuss them at the end of this subsection.
Based on the feature of these criteria, we compare the results of DeepGini with existing coverage-based methods in two groups:
(1) NAC, NBC, and SNAC; (2) TKNC, LSC, DSC, and KMNC.

\subsubsection{Comparing with NAC, NBC, and SNAC}

As shown in Table~\ref{tab:apfd_of_mnist}, DeepGini is capable of prioritizing tens of thousands of tests within 2 seconds.
The APFD value of DeepGini is very close to 1, which implies that our approach is very close to the theoretically best approach.
Table \ref{tab:apfd_of_mnist} also shows that fewer than 0.5\% of tests are sufficient to achieve the maximum coverage rate of the three coverage criteria: NAC, NBC, and SNAC, regardless of their parameter settings.
In the 10,000 original tests of MNIST, a very small amount of tests are sufficient for achieving the maximum coverage rate: for NAC(0.75), 22 tests can reach the maximum coverage (84\%); for NBC(0.5), 5 tests can reach the maximum coverage (0.97\%); and for SNAC(0.5), 5 tests can reach the maximum coverage (2\%).
Since we achieve the maximum coverage rate very quickly,
CAM will degenerate into CTM very quickly.
Thus, the effectiveness and the efficiency of CAM are almost the same as CTM for these datasets.

\textbf{Effectiveness.} Using MNIST as an example,
Figure \ref{fig:apfd} plots the number of detected misclassified tests against the prioritized tests.
We have two observations from this figure.
% First, our prioritization method can find more misclassified tests much faster than neuron-coverage-based methods.
% Second, as illustrated by the dotted lines in Figure \ref{fig:apfd}a,
% neuron-coverage-based prioritization methods, sometimes, are even worse than a random prioritization.
First, DeepGini achieves a higher APFD value in comparison to NAC, NBC, and SNAC.
Second, as illustrated by the dotted lines in Figure \ref{fig:apfd}, neuron-coverage-based prioritization methods, sometimes, are even worse than the random prioritization strategy.

\textbf{Efficiency.} 
Table~\ref{tab:apfd_of_mnist} shows that, for the original test sets, the CAM-based prioritization processes of NAC, NBC, and SNAC cost at least 2, 5, 7 seconds respectively, while DeepGini costs only 0.45 seconds. 
Similarly, for the test set with the adversarial examples, we observe that the CAM-based prioritization processes of the three baselines cost more than 11 seconds, while DeepGini costs only 2 seconds. 
This data shows that DeepGini has a higher efficiency in comparison with NAC, NBC, and SNAC. 

% Since CAM degenerates into CTM as explained above and both the CTM method and our method use quick-sort to prioritize tests, the differences between their time cost are not significant.

\subsubsection{Comparing with TKNC, LSC, DSC, and KMNC}

As discussed in Section \ref{subsec:coverage_dnn},
every single test has the same coverage rate of TKNC, LSC, and DSC, regardless of its parameter $k$.
Thus, CTM does not work if we use these coverage metrics to prioritize tests.
Unfortunately, CAM does not work well using these coverage metrics either.
The main reason is that fewer than 5\% of tests are enough to achieve the maximal coverage rate.
After prioritizing the 5\% tests, CAM is degenerate into CTM, which does not work as explained above.

Similarly, CTM does not work if we use  KMNC to prioritize tests, because almost all single tests have the same coverage rate of KMNC, regardless of its parameter $k$.
However, KMNC can work with the CAM prioritization method. 
Thus, we only compare KMNC-based CAM with our prioritization method.

% we only compare KMNC-based CAM with our prioritization method.
% , we only can randomly prioritize the remaining tests.

\textbf{Effectiveness.} 
In the example of LeNet-5 on MNIST, Figure~\ref{fig:apfd} plots the prioritization results, in which the curve of our method goes up far more quickly than the four baseline methods.
In the original test set, while DeepGini has obtained the APFD value of 0.984, TKNC, LSC, DSC, and KMNC only obtain 0.023, 0.658, 0.658, and 0.599. 
We can observe similar trends for the test set with adversarial examples.
% Furthermore, we can observe that the random prioritization strategy, which is denoted by the black curve, also outperforms the TKNC-based method.
This result implies that the DeepGini significantly outperforms the four baselines regarding the effectiveness of prioritizing tests.

\textbf{Efficiency.} 
Table \ref{tab:apfd_of_mnist} shows that prioritization methods based on these coverage metrics are 20$\times$-2000$\times$ slower than our method. 
One special case is KMNC-based CAM method.
When prioritizing tests using KMNC-based CAM method, we observe serious efficiency issues.
% That is, because the time complexity of the method is very high, we cannot finish prioritizing tests in an acceptable time budget.
% We cannot succeed in prioritizing tests using this method in 12 hours for MNIST.
That is because the time complexity of KMNC-based CAM method is very high, we cannot finish prioritizing tests in an acceptable time budget.
For MNIST, we cannot succeed in prioritizing tests using the method in 12 hours.
Considering MNIST is a relatively small dataset, the efficiency problem would be the bottleneck of applying KMNC-based CAM method in practice.

\begin{table*}[!ht]
\caption{The DNNs' accuracy value after retraining with first 10\% prioritized tests.}
	\def\arraystretch{0.9}
\begin{tabular}{lcc|cc|cc|cc|l}
\hline
                   & \multicolumn{2}{c|}{MNIST} & \multicolumn{2}{c|}{CIFAR-10} & \multicolumn{2}{c|}{FASHION}  & \multicolumn{2}{c|}{SVHN}     & \multirow{2}{*}{Avg} \\
                   & LeNet-1       & LeNet-5    & ResNet-20     & VGG-16        & LeNet-1       & ResNet-20     & LeNet-5       & VGG-16        &                      \\ \hline
NAC(0.75)-CTM      & 0.89          & 0.93       & 0.93          & 0.79          & 0.81          & 0.91          & 0.82          & 0.67          & 0.84                 \\
NAC(0.75)-CAM      & 0.83          & 0.85       & 0.92          & 0.76          & 0.78          & 0.94          & 0.8           & 0.74          & 0.83                 \\
NBC(0)-CAM         & 0.84          & 0.88       & 0.92          & 0.78          & 0.77          & 0.94          & 0.81          & 0.75          & 0.84                 \\
NBC(0)-CTM         & 0.91          & 0.84       & 0.93          & 0.75          & 0.81          & 0.95          & 0.82          & 0.75          & 0.85                 \\
SNAC(0)-CTM        & 0.84          & 0.84       & 0.93          & 0.79          & 0.83          & 0.95          & 0.82          & 0.75          & 0.84                 \\
SNAC(0)-CAM        & 0.82          & 0.87       & 0.91          & 0.76          & 0.79          & 0.94          & 0.81          & 0.76          & 0.83                 \\
LSC(1000, 100)-CAM & 0.84          & 0.86       & 0.92          & 0.81          & 0.8           & 0.94          & 0.81          & 0.74          & 0.84                 \\
DSC(1000, 100)-CAM & 0.84          & 87         & 0.92          & 0.82          & 0.81          & 0.94          & 0.82          & 0.75          & 0.85                 \\
TKNC(1)-CAM        & 0.83          & 0.88       & 0.92          & 0.8           & 0.78          & 0.94          & 0.81          & 0.8           & 0.85                 \\
DeepGini           & \textbf{0.99} & \textbf{1} & \textbf{0.98} & \textbf{0.93} & \textbf{0.89} & \textbf{0.98} & \textbf{0.93} & \textbf{0.97} & \textbf{0.96}        \\ \hline
\end{tabular}
\label{tab:retrain-accuracy}
\end{table*}

\subsection{Guidance (RQ3)}

Figure~\ref{fig:retrain} illustrates the experimental results for RQ3.
Like the previous experiments, we put the coverage metrics into two groups.
Figure~\ref{fig:retrain}-a and Figure~\ref{fig:retrain}-b demonstrate the results of LeNet-5 on MNIST.
The curves show the accuracy of the DNN after retraining with 1\%, 2\%, ..., 10\% tests.
As we introduced in Section~\ref{sec:rq}, the testing set $T$ and validation set $V$ is divided from the adversarial test data.
This setting makes initial accuracy of the DNN are the same for all metrics, i.e., it is 0 without retraining. 
% The initial accuracy of the DNN are the same for all metrics.
Note that, because KMNC-based prioritization technique is not scalable, we cannot finish the experiment of RQ3 in 12 hours.

% For the third group, i.e., the KMNC metric, because KMNC-based prioritization technique is not scalable, we cannot finish the experiment of RQ3 in 12 hours.
% As illustrated, the curves of DeepGini grows much faster than other metrics, allowing us to improve the accuracy of a DNN to above 95\% in a more efficient manner.

The curves show that retraining with the tests prioritized by DeepGini is more effective in improving the accuracy of DNNs.    
This outperformance is clear in the retraining with all sizes of the test sets. 
We also present the accuracy value after retraining the DNN with the first 10\% prioritized tests in Table~\ref{tab:retrain-accuracy}. 
From the table, we can observe that the outperformance can be found across all combinations of datasets and DNN models.  
In average, while the baseline criteria can reach 0.83 --- 0.85 accuracy, DeepGini can improve the accuracy value to 0.96, which is close to 1.

% \begin{figure}[ht]
% \centering
% \begin{subfigure}[b]{.495\linewidth}
%   \centering
%   \includegraphics[width=\linewidth]{./figs/retrain-mnist-1}
%   \caption{}
%   \label{fig:retrain-a}
% \end{subfigure}
% \begin{subfigure}[b]{.495\linewidth}
%   \centering
%   \includegraphics[width=\linewidth]{./figs/retrain-mnist-2}
%   \caption{}
%   \label{fig:retrain-b}
% \end{subfigure}%
% \caption{Enhancing the robustness of the DNN with prioritized tests (MNIST with LeNet-5).}
% \label{fig:retrain}
% \end{figure}

\begin{figure}[ht]
\centering
\begin{subfigure}[b]{.85\linewidth}
  \centering
  \includegraphics[width=\linewidth]{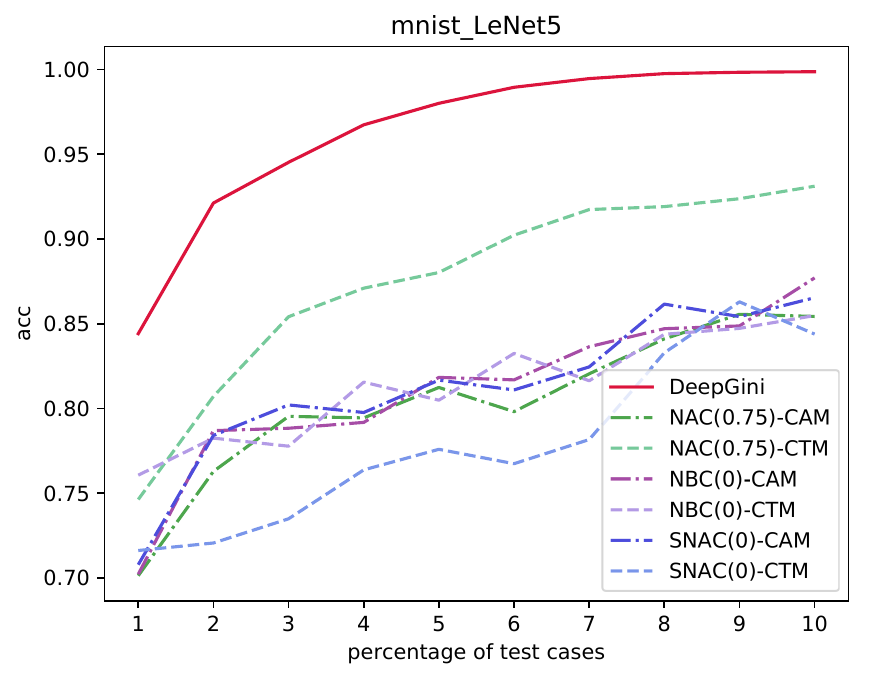}
  \caption{}
  \label{fig:retrain-a}
\end{subfigure}
\begin{subfigure}[b]{.85\linewidth}
  \centering
  \includegraphics[width=\linewidth]{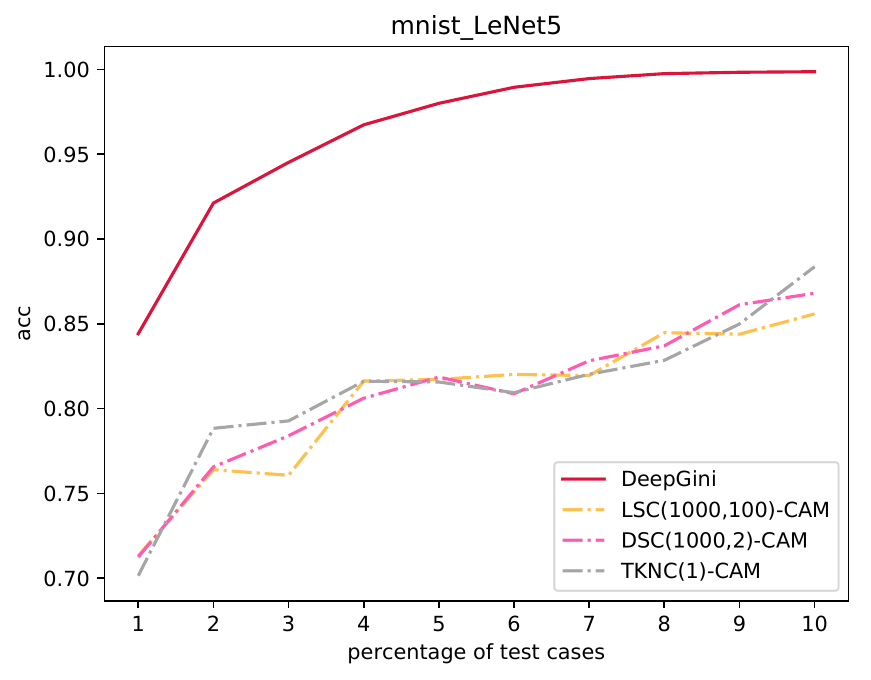}
  \caption{}
  \label{fig:retrain-b}
\end{subfigure}%
\caption{Enhancing the robustness of the DNN with prioritized tests (MNIST with LeNet-5).}
\label{fig:retrain}
\end{figure}

\subsection{Discussion}

For each of these existing neural coverage based criteria, we summarize our findings in Table~\ref{tab:results_of_prioritization}.
Although structural coverage criteria are very effective in classic testing methods,
they are not effective in the new scenario of DNN testing.
More specifically,
for some of these criteria,
we observe that a very small part, around 1\% to 5\%, of original tests are sufficient to achieve the maximal coverage.
For some of these criteria,
the coverage criteria cannot distinguish different tests.
These results extend 
the existing discussion on the effectiveness of structural coverage~\cite{li2019structural} -- 
some researchers cast doubts on the effectiveness of neuron-coverage-guided test generation techniques. 
For instance, in our experiment, the random prioritization strategy even outperforms NAC and TKNC, which implies that these coverage criteria could be misleading in finding new incorrect DNN behaviors.

Further, DeepGini is designed based on the assumption that the test produces similar probabilities for all classes has a higher probability of being misclassified. 
Admittedly, this assumption is not necessarily true. 
There could be some variance/difference in the probabilities of the different classes, yet a wrong class may be given a higher probability than the expected class.
However, the results of RQ1 shows the APFD value of DeepGini is higher than baselines, which indicates that it can efficiently detect a large number (but not all) of misclassified tests in comparison with baselines. 

On the other hand, 
test prioritization should be scalable and efficient so that it can be applied in the scenario of DNN testing.
Such scalability and efficiency requirement are very necessary because the number of tests is usually very large in DNN testing,
which is different from conventional software testing.
For instance, we found that the KMNC criteria failed to work well regarding both the effectiveness and efficiency. 
In the future, more research should be conducted to investigate its potential and improvement.  
Considering MNIST is a very basic dataset for deep learning, 
we suggest software engineers do not apply these criteria in their engineering practice before they are improved.

\begin{table*}[!ht]
\small
	\centering
	\caption{Summary of Our Findings on Test Prioritization}\label{tab:results_of_prioritization}
	\begin{tabulary}{1.0\textwidth}{r|L}
		\toprule
		\textbf{Metrics} & \textbf{Findings} \\
		\midrule
		\multirow{3}{*}{NAC/NBC/SNAC}
		&1. CAM will quickly degenerate into CTM for these metrics because only a small number of tests can achieve the maximum coverage.\\
		&2. CTM is not effective and even worse than random prioritization when we use these metrics.\\
		\hline
		
		\multirow{3}{*}{TKNC/LSC/DSC}
		&3. CAM will quickly degenerate into CTM for TKNC/LSC/DSC because only a small number of tests 
		can achieve the maximum coverage rate.         \\
		&4. CTM does not work when TKNC/LSC/DSC are used because all single tests have the same coverage rate. \\
		&5. Computing LSC/DSC additionally relies on the training set.\\
		\hline
		
		\multirow{2}{*}{KMNC}
		&6. CAM is not scalable due to its high complexity when KMNC is used. \\
		&7. CTM does not work when KMNC is used because almost all single tests have the same coverage rate.\\
		\hline
		
		\multirow{3}{*}{\textbf{DeepGini}}
		&8. DeepGini is the most effective and efficient metric for test prioritization regarding the APFD value and the time cost.\\
		&9. DeepGini does not relies on anything except for the tests and the DNN to test.\\
		\bottomrule
	\end{tabulary}
\end{table*}

\section{Related Work}
\label{sec:relwork}

We discuss the related work in two groups: 
(1) test prioritization methods for conventional software and
(2) testing techniques for deep learning systems.

\subsection{Test Prioritization Techniques}

Test prioritization seeks to find the ideal ordering of tests,
so that software testers or developers can obtain maximal benefit in a limited time budget.
The idea was first mentioned by Wong et al.~\cite{wong1998effect} and then the technique was proposed by Harrold and Rothermel~\cite{rothermel1996analyzing,harrold1999testing}
in a more general context.
We observe that such an idea from the area of software engineering can significantly reduce the effort of labeling for deep learning systems.
This is mainly because a deep learning system usually has a large number of unlabeled tests 
but developers only have limited time for labeling.

Coverage-based test prioritization, such as the CAM and CTM methods studied in this paper, is one of the most commonly studied prioritization techniques. 
In conventional software engineering, we can obtain a new prioritization method when a different coverage criterion is applied.
Rothermel et al.~\cite{rothermel1999test,rothermel2001prioritizing} reported empirical studies of several coverage-based approaches, driven by branch coverage, statement coverage, and so-called FEP, a coverage criterion inspired by mutation testing~\cite{budd1981mutation}.
In addition, 
Jones and Harrold~\cite{jones2003test} reported that MC/DC, a stricter form of branch coverage, is also applicable to coverage-based test prioritization.
Different from the above techniques, we focus on testing and debugging for deep learning systems. Thus, we studied test prioritization based on coverage criteria that specially proposed for DNNs. Our study demonstrated that, using these coverage criteria, coverage-based test prioritization is not effective and efficient. Sometimes, its effectiveness is even worse than random prioritization. Instead, our approach uses a simple metric that does not require to profile the DNNs but is effective and also efficient.

We notice that, in software engineering, there are also many prioritization techniques based on metrics other than coverage criteria, including distribution-based approach~\cite{leon2003comparison}, human-based approach~\cite{tonella2006using,yoo2009clustering}, history-based approach~\cite{sherriff2007prioritization}, model-based approach~\cite{korel2007model,korel2005test,korel2008application}, and so on. 
These techniques are specially-designed for conventional software systems instead of deep learning systems. Making them applicable to deep learning systems may require non-trivial efforts of re-design. We leave them as our future work.

\subsection{Testing Deep Learning Systems}

In conventional practice, machine learning models were mainly evaluated using available validation datasets~\cite{witten2016data}.
However, these datasets usually cannot cover various corner cases that may induce unexpected behaviors~\cite{pei2017deepxplore,tian2018deeptest}.
To further ensure the quality of a deep learning system, software-engineering researchers have designed many testing approaches.
Pei et al.~\cite{pei2017deepxplore} proposed \textit{DeepXplore}, the first white-box testing framework, to identify and generate the corner-case inputs that may induce different behaviors over multiple DNNs.
Ma et al.~\cite{ma2018deepmutation} presented a mutation testing framework for DNNs aiming at evaluating the quality of datasets.
Tian et al.~\cite{tian2018deeptest} presented \textit{DeepTest} to generate test inputs by maximizing the numbers of activated neurons via a basic set of image transformations. 
Zhang et al.~\cite{zhang2018deeproad} employed generative adversarial network to transform the driving scenes into various weather conditions, which increases the diversity of datasets.
Different from the above techniques that rely on solid test oracle, 
our method focuses on the problem that we usually have a large number of tests without test oracle.
We observe that the idea of test prioritization can enable developers to obtain as many misclassified tests as possible in a limited time budget,
thereby easing the burden of labeling.
However, in comparison with traditional software programs, the modern DNNs often consist of millions of neurons and hundreds of layers, which naturally enlarges its potential testing space.
While the sophisticated internal logic of a DNN makes it challenging to adopt the idea of coverage criteria to test prioritization,
this paper introduces a new metric that only analyzes the output space of a DNN and is able to effectively guide the test prioritization.
We notice that researchers have proposed some preliminary prioritization methods for testing DNNs~\cite{zhang2019noise,byun2019input}.
They use different methods to prioritize the tests but failed to compare their techniques with classic coverage-based methods.
Furthermore, they did not provide any information on the capability of enhancing the DNN robustness.

To guide the testing techniques for DNNs,
Pei et al.~\cite{pei2017deepxplore} introduced neuron activation coverage to measure the differences of the execution of test data. 
Ma et al.~\cite{ma2018deepgauge} designed a set of multi-level and multi-granularity testing criteria for assessing the quality of testing of deep learning systems.
Our approach has shown that it is not effective or efficient to prioritize tests based on these coverage criteria.
Sun et al.~\cite{sun2018testing,sun2018concolic} presented a concolic testing framework that incrementally generates a set of test inputs to improve coverage by alternating between concrete execution and symbolic analysis.
The MC/DC-like coverage criteria proposed in the paper~\cite{sun2018testing,sun2018concolic} does not work for prioritization because of two reasons. First, since we need at least a pair of tests to compute the coverage rate, the coverage rate of a single test is meaningless. Thus, CTM does not work. Second, computing the coverage rate is of at least quadratic complexity. Thus, it is not scalable, just like the KMNC coverage shown in our evaluation. Since we only focus on coverage criteria published in peer-reviewed venues, these MC/DC-like metrics are not included but just briefly discussed here.
Different from such test generation techniques that also have the oracle problem, 
our approach attempts to prioritize tests so that the oracle problem is alleviated.

\section{Conclusion}
\label{sec:conc}

Based on a statistical view of DNN,
we have introduced an approach, namely DeepGini, to prioritizing testing data so that
we can improve the quality of DNN efficiently.
Experimental results demonstrate that it is more effective and efficient than coverage-based methods.
In real-world scenario, tests usually do not have labels and
we have to invest a lot of manpower to label them.
With such a prioritization method in hand,
we can achieve maximal benefit, even the labeling process is prematurely halted at some arbitrary point due to resource limits.

\section*{Acknowledgements}
We would like to thank anonymous reviewers for their insightful comments. This project was partially funded by the National Natural Science Foundation of China under Grant Nos. 61832009 and 61932012.
Qingkai Shi is the corresponding author.

\balance
\bibliographystyle{ACM-Reference-Format}
\bibliography{ref}

\end{document}